\documentclass{article}

\usepackage{amsfonts}
\newcommand{\Cite}[1]{\cite{#1}}

\renewenvironment{abstract}[1]%
  {\begin{quotation}\begin{small}}%
  {\end{small}\end{quotation}}
\newenvironment{acknowledgment}[1]%
  {\paragraph{Acknowledgments}}{}

\newcommand{\REVISED}[1]{#1}
\usepackage{amsmath,amssymb}
\usepackage[dvipdfm]{graphicx}
\usepackage[usenames,dvipdfm]{color}

\newcommand{\comment}[1]{}

\def\imply{\Rightarrow}
\def\gsim{\gtrsim}
\def\lsim{\lesssim}
\def\da{{\mathcal{\downarrow}}}

\def\B{{\cal B}}
\def\C{{\cal C}}
\def\D{{\cal D}}
\def\S{{\cal S}}
\def\T{{\cal T}}
\def\U{{\cal U}}
\def\V{{\cal V}}

\def\args{\mathit{args}}
\def\safe{\mathit{safe}}

\def\type{\mathit{type}}
\def\sub{\mathit{sub}}

\def\least{\mathrm{least}}

\newcommand{\sym}[1]{\mathrm{#1}}
\newcommand{\ol}[1]{\overline{#1}}
\newcommand{\ul}[1]{\underline{#1}}
\newcommand{\pair}[1]{\langle\,#1\,\rangle}

\newcommand{\Reduction}[9]{\mathrel{%
  {}_{#1}^{#2}%
  \raise#7\hbox{%
    \vtop{\ialign{##\crcr%
      \hfil\raise#8\hbox{$\scriptstyle{\ #4\ }$}\hfil\crcr%
      \noalign{\nointerlineskip}%
      #9\crcr%
      \noalign{\nointerlineskip}%
      $\hfil\scriptscriptstyle{\ #3\ }\hfil$\crcr%
    }}%
  }%
  {}_{\scriptscriptstyle #5}^{#6}}}%

\newcommand{\red}[2]{%
  \Reduction{}{}{#1}{#2}{}{}{0.8ex}{0.2ex}{\rightarrowfill}}

\newtheorem{definition}{Definition}[section]
\newtheorem{lemma}[definition]{Lemma}
\newtheorem{theorem}[definition]{Theorem}

\newtheorem{proposition}[definition]{Proposition}
\newtheorem{example}[definition]{Example}
\newtheorem{proof}{Proof.}

\def\qed{\hspace*{\fill}$\square$}

\begin{document}


%
%

\begin{center}
{\Large \bf
  Argument Filterings and Usable Rules\\[1mm] in Higher-Order Rewrite Systems}\\
\vspace*{20pt}
{\large
  SUZUKI Sho${}^\dag$,
  KUSAKARI Keiichirou${}^\dag$,\\[1mm]
  Fr\'{e}d\'{e}ric BLANQUI${}^\ddag$}\\
\vspace*{10pt}
  $\dag$ Graduate School of Information Science, Nagoya University\\\vspace*{1mm}
  $\ddag$ INRIA, France
\end{center}

\begin{abstract}\ 
The static dependency pair method is a method for proving the termination
of higher-order rewrite systems {\em \`a la} Nipkow.
It combines the dependency pair method introduced
for first-order rewrite systems
with the notion of strong computability
introduced for typed $\lambda$-calculi.
Argument filterings and usable rules are two important methods
of the dependency pair framework
used by current state-of-the-art first-order automated termination provers.
In this presentation,
we extend the class of higher-order systems
on which the static dependency pair method can be applied.
Then, we extend argument filterings and usable rules
to higher-order rewriting,
hence providing the basis for a powerful automated termination prover
for higher-order rewrite systems.
\end{abstract}


\section{Introduction}

Various extensions of term rewriting systems (TRSs) \cite{T03}
for handling functional variables and abstractions have been proposed
\cite{klop80phd,N91,jouannaud91lics,oostrom94phd,K01}.
In this paper,
we consider higher-order rewrite systems (HRSs) \cite{N91}, that is,
rewriting on $\beta$-normal $\eta$-long simply-typed $\lambda$-terms
using higher-order matching.

For example, the typical higher-order function
{\em foldl} can be defined by the following HRS:
\REVISED{%
\[R_{\sym{foldl}} = \left\{\begin{array}{rcl}
  \sym{foldl}(\lambda xy.F(x,y), X, \sym{nil}) &\to& X \\
  \sym{foldl}(\lambda xy.F(x,y), X, \sym{cons}(Y,L))
    &\to& \sym{foldl}(\lambda xy.F(x,y), F(X,Y), L)
\end{array}\right.\]
Here we suppose that the function {\em foldl} has the type
$(\mathbb N \to \mathbb N \to \mathbb N) \to L \to \mathbb N$,
and $L$ is a type of natural number's list.
}
Then, the functions {\em sum} and {\em len},
computing the sum of the elements and the number of elements respectively,
can be defined by the following HRSs:
\begin{eqnarray*}
  R_{\sym{sum}}
    &=& R_{\sym{foldl}} \cup \left\{\begin{array}{rcl}
          \sym{add}(0,Y)          &\to& Y \\
          \sym{add}(\sym{s}(X),Y) &\to& \sym{s}(\sym{add}(X,Y)) \\
          \sym{sum}(L)            &\to& \sym{foldl}(\lambda xy.\sym{add}(x,y),0,L)
        \end{array}\right. \\
  R_{\sym{len}}
    &=& R_{\sym{foldl}} \cup \left\{\begin{array}{rcl}
          \sym{len}(L) &\to& \sym{foldl}(\lambda xy.\sym{s}(x),0,L)
        \end{array}\right.
\end{eqnarray*}
In the HRS $R_{\sym{len}}$,
the anonymous function $\lambda xy.\sym{s}(x)$
is represented by using $\lambda$-abstraction.

The static dependency pair method is a method
for proving the termination of higher-order rewrite systems.
It combines the dependency pair method
introduced for first-order rewrite systems \cite{AG00}
with Tait and Girard's notion of strong computability
introduced for typed $\lambda$-calculi \cite{girard88book}.
It was first introduced
for simply-typed term rewriting systems (STRSs) \cite{KS07}
and then extended to HRSs \cite{KISB09}.
The static dependency pair method consists
in showing the non-loopingness of each static recursion component independently,
the set of static recursion components being computed
through some static analysis of the possible sequences of function calls.

This method applies only to plain function-passing (PFP) systems.
In this paper,
we provide a new definition of PFP
that significantly enlarges the class of systems
on which the method can be applied.
It is based on the notion of accessibility
introduced in \Cite{blanqui02tcs} and extended to HRSs in \Cite{B00}.

For the HRS $R_{\sym{sum}} \cup R_{\sym{len}}$,
the static dependency pair method returns the following two components:
\[\begin{array}{l}
    \left\{\begin{array}{l}
      \sym{foldl}^\sharp(\lambda xy.F(x,y), X, \sym{cons}(Y,L))
        \to \sym{foldl}^\sharp(\lambda xy.F(x,y), F(X,Y), L)
    \end{array}\right\} \\
    \left\{\begin{array}{l}
      \sym{add}^\sharp(\sym{s}(X),Y) \to \sym{add}^\sharp(X,Y)
    \end{array}\right\}
\end{array}\]
The static dependency pair method proves
the termination of the HRS $R_{\sym{sum}} \cup R_{\sym{len}}$
by showing the non-loopingness of each component.

In order to show the non-loopingness of a component,
the notion of reduction pair is often used.
Roughly speaking, it consists in finding a well-founded quasi-ordering
in which the component rules are strictly decreasing
and all the original rules are non-increasing.

Argument filterings, which consist in removing
some arguments of some functions, provide a way
to generate reduction pairs. First introduced for TRSs \cite{AG00},
it has been extended to STRSs \cite{K01,KS09}.
In this paper, we extend it to HRSs.

In order to reduce the number of constraints
required for showing the non-loopingness of a component,
the notion of usable rules is also very important.
Indeed, a finer analysis of sequences of function calls show that
not all original rules need to be taken into account when trying to prove
the termination of a component.
This analysis was first conducted for TRSs \cite{GTSF06,HM07}
and has been extended to STRSs \cite{SKSSN07,KS09}.
In this paper, we extend it to HRSs.

All together, this paper provides a strong theoretical basis for the
development of an automated termination prover for HRSs, by extending
to HRSs some successful techniques used by modern state-of-the-art
first-order termination provers like for instance
\cite{giesl06ijcar,HM07}.

The remainder of this paper is organized as follows.
Section 2 introduces HRSs.
Section 3 presents the static dependency pair method
and extend the class of systems on which it can be applied.
In Section 4,
we extend the argument filtering method to HRSs.
In Section 5,
we extend the notion of usable rules on HRSs.
Concluding remarks are given in Section 6.

\section{Preliminaries}

In this section,
we introduce the basic notions for HRSs
according to \Cite{N91,MN98}.

The set $\S$ of {\em simple types} is generated
from the set $\B$ of {\em basic types} by the type constructor $\to$.
A {\em functional} or {\em higher-order type}
is a simple type of the form $\alpha \to \beta$.
We denote by $\rhd_s$ the strict subterm relation on types.

A {\em preterm} is generated from an infinite set of typed variables $\V$
and a set of typed function symbols $\Sigma$ disjoint from $\V$
by $\lambda$-abstraction and $\lambda$-application.
The set of typed preterms is denoted with $\T^{pre}$.
We denote by $t\da$ the $\eta$-long $\beta$-normal form
of a simply-typed preterm $t$.
The set $\T$ of {\em (simply-typed) terms}
is defined as $\{ t\da \mid t \in \T^{pre} \}$.
The unique type of a term $t$ is denoted by $\type(t)$.
We write $\V_\alpha$ (resp. $\T_\alpha$) as the set of variables
(resp. terms) of type $\alpha$,
The $\alpha$-equivalence of terms is denoted by $\equiv$.
The set of free variables in a term $t$ is denoted by $FV(t)$.
We assume for convenience that bound variables in a term are all different,
and are disjoint from free variables.
In general, a term $t$ is of the form $\lambda x_1\ldots x_m. a t_1\ldots t_n$
where $a \in \Sigma \cup \V$.
We abbreviate this by $\lambda \ol{x_m}.a(\ol{t_n})$.
For a term $t \equiv \lambda \ol{x_m}. a(\ol{t_n})$,
the symbol $a$, denoted by $top(t)$, is the {\em top symbol} of $t$,
and the set $\{\ol{t_n}\}$, denoted by $\args(t)$,
is the {\em arguments} of $t$.
We define the set $Sub(t)$ of {\em subterms} of $t$
by $\{t\} \cup Sub(s)$ if $t \equiv \lambda x.s$, and
$\{t\} \cup \bigcup_{i=1}^n Sub(t_i)$ if $t \equiv a(\ol{t_n})$.
We use $t \unrhd_\sub s$ to represent $s \in Sub(t)$,
and define $t \rhd_\sub s$ by $t \unrhd_\sub s$ and $t \not\equiv s$.
The set $Pos(t)$ of {\em positions} in a term $t$
is the set of strings over positive integers
inductively defined as
$Pos(\lambda x. t) = \{ \varepsilon \} \cup \{ 1p \mid p \in Pos(t) \}$
and
$Pos(a(\ol{t_n})) = \{\varepsilon\} \cup \bigcup_{i=1}^n \{ip \mid p \in Pos(t_i) \}$.
The {\em prefix order} $\prec$ on positions is defined
by $p \prec q$ iff $pw = q$ for some $w\neq \varepsilon$.
The subterm of $t$ at position $p$ is denoted by $t|_p$.

A term containing a unique occurrence of the special constant $\square_\alpha$ of type $\alpha$
is called a {\em context}, denoted by $C[\,]$.
We use $C[t]$ for the term obtained from $C[\,]$
by replacing $\square_\alpha$ with $t \in \T_\alpha$.
A substitution $\theta$ is a mapping from variables to terms
such that $\theta(X)$ has the type of $X$ for each variable $X$.
We define $dom(\theta) = \{ X \mid X\da \not\equiv \theta(X) \}$
and assume that $dom(\theta)$ is always finite.
A substitution $\theta$ is naturally extended to a mapping from terms to terms.
We use $t\theta$ instead of $\theta(t)$ in the remainder of the paper.
A substitution $\theta$ is said to be a {\em variable permutation}
if $\forall X \in dom(\theta). \exists Y \in dom(\theta).
\theta(X) \equiv Y\da$ and $\theta(X) \equiv \theta(Y) \imply X=Y$ hold.

Following \Cite{MN98}, a
{\em higher-order rewrite rule} is a pair $(l,r)$ of terms,
denoted by $l \to r$,
such that $top(l) \in \Sigma$, $\type(l) = \type(r) \in \B$
and $FV(l) \supseteq FV(r)$. Since, by definition,
terms are in $\eta$-long form, function symbols are always applied
to the same (maximal) number of arguments. Considering non-$\eta$-normal terms
or rules of functional type is outside the scope of this paper.
An HRS is a set of higher-order rewrite rules.
The {\em reduction relation} $\red{R}{}$ of an HRS $R$
is defined by $s \red{R}{} t$
iff $s \equiv C[l\theta\da]$ and $t \equiv C[r\theta\da]$
for some rewrite rule $l \to r \in R$, context $C[\,]$
and substitution $\theta$.
The transitive and reflexive-transitive closures of $\red{R}{}$
are denoted by $\red{R}{+}$ and $\red{R}{*}$, respectively.
An HRS $R$ is said to be {\em finitely branching}
if $\{ t' \mid t \red{R}{} t' \}$ is a finite set for any term $t$.

A term $t$ is said to be
{\em terminating} or {\em strongly normalizing} for an HRS $R$,
denoted by $SN(R,t)$,
if there is no infinite rewrite sequence of $R$ starting from $t$.
We write $SN(R)$ if $SN(R,t)$ holds for any term $t$.
A well-founded relation $>$ on terms is a {\em reduction order}
if $>$ is closed under substitution and context.
We notice that an HRS $R$ is terminating
iff $R \mathrel{\subseteq} {>}$ for some reduction order $>$.

A term $t$ is said to be {\em strongly computable} in an HRS $R$
if $SC(R,t)$ holds,
  which is inductively defined on simple types as follows:
  $SN(R,t)$ if $type(t) \in \B$,
  and $\forall u \in \T_\alpha. (SC(R, u) \imply SC(R, (t u)\da))$
  if $type(t)=\alpha \to \beta$.
We also define the set
  $\T_{SC}^{\args}(R) = \{ t \mid \forall u \in \args(t).SC(R,u) \}$.

Finally, we introduce the proposition required for later proof.
\begin{proposition}\cite{MN98}\label{prop:red-subst}
If $s \red{R}{*} t$ and $\theta \red{R}{*} \theta'$
(i.e. $\forall x \in \V. x\theta \red{R}{*} x\theta'$)
then $s\theta\da \red{R}{*} t\theta'\da$.
\end{proposition}

\section{Improved Static Dependency Pair Method}

In this section,
we introduce the static dependency pair method
for plain function-passing (PFP) HRSs \cite{KISB09}
but extend the class of PFP systems.

\REVISED{%
The method in \Cite{KISB09} applies only to PFP systems.
From a technical viewpoint,
we have noticed that
the unclosedness of strong computability with respect to the subterm relation
is the reason why the method is not applicable to every HRS.
Hence we can extend the applicable class for the method
if more strongly computable subterms can be acquired.
From the same motivation,
Blanqui introduced the notion of accssibility 
to design a higher-order path ordering \cite{B00}.
By using the notion of accessibility,
we provide a new definition of PFP
that enlarges the class of systems on which the method can be applied.
}


\begin{definition}[Stable subterms]
The {\em stable subterms} of $t$ are $SSub(t)=SSub_{FV(t)}(t)$ where
$SSub_X(t)=\{t\}\cup SSub'_X(t)$, $SSub'_X(\lambda x.s)=SSub_X(s)$,
$SSub'_X(a(\ol{t_n}))=\bigcup_{i=1}^nSSub_X(t_i)$ if $a\notin X$, and
$SSub'_X(t)=\emptyset$ otherwise.
\end{definition}

\begin{lemma}\label{lem-stable}
(1) $SSub(t)\subseteq Sub(t)$. (2) If $u\in SSub(t)$ and
  $dom(\theta)\subseteq FV(t)$, then $u\theta\da\in
  SSub(t\theta\da)$. (3) If $u\in Sub(t)$ and $t\in SN$, then $u\in
  SN$.
\end{lemma}


\begin{definition}[Safe subterms - New definition]\label{def-safe}
The set of {\em safe subterms} of a term $l$ is
$\safe(l)=\bigcup_{l'\in \args(l)}\{t\da\mid t\in Acc(l'),
FV(t)\subseteq FV(l')\}$
where $t\in Acc(l')$ ($t$ is {\em accessible} in $l'$) if either:
\begin{enumerate}
\setcounter{enumi}{-1}
\item\label{acc-arg}
\REVISED{%
  $t=l'$,
}
\item\label{acc-sub}
\REVISED{%
  $t\in SSub(l')$, $\type(t)\in\B$ and $FV(t)\subseteq FV(l')$,
}
\item\label{acc-abs}
$\lambda x.t\in Acc(l')$ and $x\notin FV(l')$,
\item\label{acc-eta}
$t(x\da)$ $\in Acc(l')$ and $x\notin FV(t)\cup FV(l')$,
\item\label{acc-fun}
$f(\ol{t_n})\in Acc(l')$, $t_i=\lambda\ol{x_k}.t$,
  $\type(t)\in\B$ and $\{\ol{x_k}\}\cap FV(t)=\emptyset$,
\item\label{acc-proj}
$x(\ol{t_n})\in Acc(l')$, $t_i=t$ and $x\notin FV(\ol{t_n})\cup
  FV(l')$.
\end{enumerate}
\end{definition}

Strictly speaking, $\safe(l)$ may not be included in $Sub(l)$ and,
because of (\ref{acc-eta}), accessible terms are $\beta$-normal
preterms not necessarily in $\eta$-long form.


\begin{definition}[Plain Function-Passing \cite{KISB09}]
An HRS $R$ is {\em plain function-passing} (PFP)
if for any $l \to r \in R$ and $Z(\ol{r_n}) \in Sub(r)$ such that $Z \in FV(r)$,
there exists $k\leq n$ such that $Z(\ol{r_k})\da \in \safe(l)$.
\end{definition}

For example, the HRS $R_{\sym{foldl}}$ displayed in the introduction is PFP,
because
\(\safe(\sym{foldl}(\lambda xy.F(x,y), X, \sym{cons}(Y, L))) = \{ \lambda xy.F(x,y), X, \sym{cons}(Y, L), Y, L\}\)
and
\(F\da \equiv \lambda xy.F(x,y) \in \safe(\sym{foldl}(\lambda xy.F(x,y), X, \sym{cons}(Y, L)))\).

The definition of safeness given in \Cite{KISB09} corresponds to
\REVISED{cases (\ref{acc-arg}) and (\ref{acc-sub})}.
This new definition therefore includes much more
terms, mainly higher-order patterns \cite{miller89elp}. This greatly
increases the class of rules that can be handled and the applicability
of the method since it reduces the number of dependency pairs.

For instance, the new definition allows us to handle the following rule:
\[D(\lambda x.\sym{sin}(Fx))y\to D(\lambda x.Fx)y\times\sym{cos}(Fy)\]
Indeed,
$l'=\lambda x.\sym{sin}(Fx)\in Acc(l')$ by (\ref{acc-arg}),
$\sym{sin}(Fx)\in Acc(l')$ by (\ref{acc-abs}),
$Fx\in Acc(l')$ by (\REVISED{\ref{acc-fun}})
and $F\in Acc(l')$ by (\ref{acc-eta}).
Therefore, $\safe(l)=\{l',\lambda x.Fx,y\}$.
With the previous definition, we had $\safe(l)=\{l',y\}$ only.

\REVISED{%
Also, the new definition allows us to handle the following rule:
\[\forall (\lambda x. (P x \land Q x))
  \to \forall (\lambda x. P x) \land \forall (\lambda x. Q x)\]
Indeed,
$l' = \lambda x. (P x \land Q x) \in Acc(l')$ by (\ref{acc-arg}),
$P x \land Q x \in Acc(l')$ by (\ref{acc-abs}),
$P x, Q x \in Acc(l')$ by (\ref{acc-fun}),
and $P, Q \in Acc(l')$ by (\ref{acc-eta}).
Therefore, $\safe(l)=\{l', \lambda x. P x, \lambda x. Q x \}$.
With the previous definition, we had $\safe(l)=\{ l' \}$ only.
}

For the results presented in \Cite{KISB09} to still hold, it suffices
to check that this new definition of safeness still preserves
\REVISED{strong} computability (Lemma 4.3 in \Cite{KISB09}). This can be shown by
following the proof of Lemma 10 in \Cite{B00}.


\begin{lemma}\label{lem:acc}
Let $R$ be an HRS and $l\to r\in R$.
Then
$l\theta\da\in\T_{SC}^{args}(R)$ implies $SC(R, t\theta\da)$
for any $t \in \safe(l)$ \REVISED{and substitution $\theta$}.
\end{lemma}

\begin{proof}
We first prove that $t\theta\da$ is \REVISED{strongly} computable
whenever $t\in Acc(l')$, $l'\theta\da$ is \REVISED{strongly}
computable, and $x\theta$ is \REVISED{strongly} computable
\REVISED{for any} $x\in FV(t)\setminus FV(l')$. Wlog we can assume
that $dom(\theta)\subseteq FV(t)$.
\REVISED{We prove the claim by induction on the definition of $Acc$}.

\begin{enumerate}
\setcounter{enumi}{-1}
\item
Immediate.
\item
Since $l'\theta\da$ is \REVISED{strongly} computable,
$l'\theta\da$ is strongly normalizing.
\comment{
Since $FV(t)\subseteq FV(l')$,
we can assume that $dom(\theta)\subseteq FV(l')$ wlog.
Hence, by Lemma}
By Lemma \ref{lem-stable}, $t\theta\da\in Sub(l'\theta\da)$ and $t\theta\da$ is
SN. Therefore, since $type(t)\in\B$, $t\theta\da$ is \REVISED{strongly} computable.
\item
By definition of computability.
\item
\REVISED{We have $type(t)=\alpha\to\beta$. So, let $u\in\T_\alpha$
  strongly computable and $\theta'=\theta\uplus\{x\mapsto u\}$
  ($x\notin dom(\theta)$ since $x\notin FV(t)$). Since $x\notin
  FV(t)$, we have $(t\theta\da u)\da=(t(x\da))\theta'\da$. By IH,
  $(t(x\da))\theta'\da$ is strongly computable. Therefore,
  $t\theta\da$ is strongly computable.}
\item
Since \REVISED{strong} computability on base types is equivalent to SN and
$\{\ol{x_k}\}\cap FV(t)\!=\!\emptyset$.
\item
\REVISED{The term $p_i=\lambda\ol{y_n}.y_i$ can easily be proved
  strongly computable. Then, let $\theta'=\theta\uplus\{x\mapsto
  p_i\}$ ($x\notin dom(\theta)$ since $x\notin FV(\ol{t_n})$). Since
  $x\notin FV(t_i)$, we have
  $(x(\ol{t_n}))\theta'\da=t_i\theta\da$. By induction hypothesis,
  $(x(\ol{t_n}))\theta'\da$ is strongly computable. Therefore,
  $t\theta\da=t_i\theta\da$ is strongly computable.}
\end{enumerate}

\noindent
Let now $u\in\safe(l)$. We have $u \equiv t\da$ for some $t\in
Acc(l')$ and $l'\in \args(l)$ with $FV(t)\subseteq FV(l')$. The term
$l'\theta\da$ is \REVISED{strongly} computable since
$l\theta\da\in\T_{SC}^{args}(R)$. Since $FV(t)\subseteq FV(l')$, there
is no $x\in FV(t)\setminus FV(l')$. Therefore, $u\theta\da \equiv
t\theta\da$ is \REVISED{strongly} computable.\qed
\end{proof}

This definition of safeness can be further improved (in case
\ref{acc-fun}) by using more complex interpretations for base types
than just the set of strongly normalizing terms, but this requires to
check more properties\REVISED{\cite{BJR08}}.
We leave this for future work.


We now recall the definitions of static dependency pair, static
recursion component and reduction pair, and the basic theorems
concerning these notions, including the subterm criterion
\cite{KISB09}.

\begin{definition}[Static dependency pair \cite{KISB09}]\label{def-dp}
Let $R$ be an HRS.
All top symbols of the left-hand sides of rewrite rules,
denoted by $\D_R$, are called {\em defined symbols}.
\comment{whereas all other function symbols, denoted by $\C_R$, are {\em constructors}.}

We define the {\em marked} term $t^\sharp$
by $f^\sharp(\ol{t_n})$ if $t$ has the form $f(\ol{t_n})$ with $f \in \D_R$;
otherwise $t^\sharp \equiv t$.
Then, let $\D_R^\sharp=\{ f^\sharp \mid f \in \D_R \}$.

We also define the set of {\em candidate subterms} as follows:
\(Cand(\lambda \ol{x_m}.a(\ol{t_n}))
  = \{ \lambda \ol{x_m}.a(\ol{t_n}) \}
    \cup \bigcup^{n}_{i=1} Cand(\lambda \ol{x_m}. t_i)\).

Now, a pair $\pair{l^{\sharp},\ a^{\sharp}(\ol{r_n})}$,
denoted by $l^{\sharp} \to a^{\sharp}(\ol{r_n})$,
is said to be a {\em static dependency pair} in $R$
if there exists $l \to r \in R$ such that
$\lambda \ol{x_m}. a(\ol{r_n}) \in Cand(r)$,
$a \in \D_R$, and
$a(\ol{r_k}) \da$ $\notin \safe(l)$ for all $k\leq n$.
We denote by $SDP(R)$ the set of static dependency pairs in $R$.
\end{definition}

\begin{example} \label{ex:SDP}
Let $R_{\sym{ave}}$ be the following PFP-HRS:
\REVISED{%
\[R_{\sym{ave}} = R_{\sym{sum}} \cup R_{\sym{len}} \cup \left\{\begin{array}{rcl}
  \sym{sub}(X,0)                    &\to& X \\
  \sym{sub}(0,Y)                    &\to& 0 \\
  \sym{sub}(\sym{s}(X),\sym{s}(Y))  &\to& \sym{sub}(X,Y)  \\
  \sym{div}(0,\sym{s}(Y))           &\to& 0 \\
  \sym{div}(\sym{s}(X),\sym{s}(Y))  &\to& \sym{s}(\sym{div}(\sym{sub}(X,Y),\sym{s}(Y))) \\
  \sym{ave}(L)                      &\to& \sym{div}(\sym{sum}(L),\sym{len}(L))
\end{array}\right.\]
}
Then, the set $SDP(R_\sym{ave})$ consists of the following eleven pairs:
\[\left\{\begin{array}{rcl}
    \sym{foldl}^\sharp(\lambda xy.F(x,y), X, \sym{cons}(Y,L))
        &\to& \sym{foldl}^\sharp(\lambda xy.F(x,y), F(X,Y), L) \\
    \sym{add}^\sharp(\sym{s}(X),Y)
        &\to&   \sym{add}^\sharp(X,Y) \\
    \sym{sum}^\sharp(L)
        &\to&   \sym{foldl}^\sharp(\lambda xy.\sym{add}(x,y), 0, L) \\
    \sym{sum}^\sharp(L)
        &\to&   \sym{add}^\sharp(x,y) \\
    \sym{sub}^\sharp(\sym{s}(X),\sym{s}(Y))
        &\to&   \sym{sub}^\sharp(X,Y) \\
    \sym{div}^\sharp(\sym{s}(X),\sym{s}(Y))
        &\to&   \sym{div}^\sharp(\sym{sub}(X,Y),\sym{s}(Y)) \\
    \sym{div}^\sharp(\sym{s}(X),\sym{s}(Y))
        &\to&   \sym{sub}^\sharp(X,Y) \\
    \sym{len}^\sharp(L)
        &\to&   \sym{foldl}^\sharp(\lambda xy.\sym{s}(x),0,L) \\
    \sym{ave}^\sharp(L)
        &\to&   \sym{div}^\sharp(\sym{sum}(L),\sym{len}(L)) \\
    \sym{ave}^\sharp(L)
        &\to&   \sym{sum}^\sharp(L) \\
    \sym{ave}^\sharp(L)
        &\to&   \sym{len}^\sharp(L) 
\end{array}\right.\]
\end{example}


\begin{definition}[Static dependency chain \cite{KISB09}]
Let $R$ be an HRS.
A sequence $u_0^\sharp \to v_0^\sharp, u_1^\sharp \to v_1^\sharp, \ldots$
 of static dependency pairs is a {\em static dependency chain} in $R$
 if there exist $\theta_0,\theta_1,\ldots$ such that
 $v_i^\sharp\theta_i\da \red{R}{*} u_{i+1}^\sharp\theta_{i+1}\da$
 and $u_i\theta_i\da, v_i\theta_i\da \in \T_{SC}^{\args}(R)$ for all $i$.
\end{definition}

Note that, for all $i$, $u_i^\sharp\theta_i$ and $v_i^\sharp\theta_i$
are terminating, since strong computability implies termination.

\begin{proposition}\cite{KISB09}
Let $R$ be a PFP-HRS.
If there exists no infinite static dependency chain then $R$ is terminating.
\end{proposition}

\REVISED{%
\begin{proof}
By using Lemma \ref{lem:acc} instead of Lemma 4.3 in \Cite{KISB09},
the proof of the correspondence theorem (Theorem 5.23 in \Cite{KISB09})
still holds.
\qed
\end{proof}
}


\begin{definition}[Static recursion component \cite{KISB09}]
Let $R$ be an HRS.
The {\em static dependency graph} of $R$ is the directed graph
 in which nodes are $SDP(R)$ and there exists an arc
 from $u^\sharp \to v^\sharp$ to $u'^\sharp \to v'^\sharp$
 if the sequence $u^\sharp \to v^\sharp,\ u'^\sharp \to v'^\sharp$ is a static dependency chain.

A {\em static recursion component} is a set of nodes
 in a strongly connected subgraph of the static dependency graph of $R$.
We denote by $SRC(R)$ the set of static recursion components of $R$.

A static recursion component $C$ is {\em non-looping}
 if there exists no infinite static dependency chain
 in which only pairs in $C$ occur
 and every $u^\sharp \to v^\sharp \in C$ occurs infinitely many times.
\end{definition}

\begin{proposition}\cite{KISB09} \label{pr:SDP1}
Let $R$ be a PFP-HRS such that there exists no infinite path
in the static dependency graph.
If all static recursion components are non-looping, then $R$ is terminating.
\end{proposition}

\begin{figure}[t]
\begin{center}
\begin{picture}(350,166)(0,0)

\put(0,150){\framebox(145,16){
  $\sym{ave}^\sharp(L) \to \sym{div}^\sharp(\sym{sum}(L),\sym{len}(L))$}}
\put(145,158){\vector(1,0){20}}
\put(5,150){\line(0,-1){22}}
\put(5,128){\vector(1,0){10}}
\put(165,150){\framebox(185,16){
  $\sym{div}^\sharp(\sym{s}(X),\sym{s}(Y))  \to \sym{div}^\sharp(\sym{sub}(X,Y),\sym{s}(Y))$}}
\put(176,150){\vector(-3,-2){21}}
\put(320,150){\line(0,-1){20}}
\put(320,130){\line(1,0){20}}
\put(340,130){\vector(0,1){20}}

\put(15,120){\framebox(140,16){
  $\sym{div}^\sharp(\sym{s}(X),\sym{s}(Y))  \to \sym{sub}^\sharp(X,Y)$}}
\put(155,128){\vector(1,0){15}}
\put(170,120){\framebox(140,16){
  $\sym{sub}^\sharp(\sym{s}(X),\sym{s}(Y))  \to \sym{sub}^\sharp(X,Y)$}}
\put(275,120){\line(0,-1){20}}
\put(275,100){\line(1,0){20}}
\put(295,100){\vector(0,1){20}}

\put(40,90){\framebox(90,16){
  $\sym{ave}^\sharp(L) \to \sym{sum}^\sharp(L)$}}
\put(130,98){\vector(1,0){25}}
\put(40,98){\line(-1,0){20}}
\put(20,98){\vector(0,-1){22}}
\put(155,90){\framebox(100,16){
  $\sym{sum}^\sharp(L)   \to \sym{add}^\sharp(x,y)$}}
\put(200,90){\line(0,-1){22}}
\put(200,68){\vector(1,0){20}}

\put(0,60){\framebox(170,16){
  $\sym{sum}^\sharp(L)
    \to \sym{foldl}^\sharp(\lambda xy.\sym{add}(x,y), 0, L)$}}
\put(170,68){\line(1,0){20}}
\put(190,68){\vector(0,-1){22}}
\put(220,60){\framebox(130,16){
  $\sym{add}^\sharp(\sym{s}(X),Y)    \to \sym{add}^\sharp(X,Y)$}}
\put(315,76){\line(0,1){20}}
\put(315,96){\line(1,0){20}}
\put(335,96){\vector(0,-1){20}}

\put(50,30){\framebox(300,16){
  $\sym{foldl}^\sharp(\lambda xy.F(x,y),X,\sym{cons}(Y,L))
    \to \sym{foldl}^\sharp(\lambda xy.F(x,y),F(X,Y),L)$}}
\put(300,30){\line(0,-1){20}}
\put(300,10){\line(1,0){20}}
\put(320,10){\vector(0,1){20}}

\put(0,0){\framebox(85,16){
  $\sym{ave}^\sharp(L) \to \sym{len}^\sharp(L)$}}
\put(85,8){\vector(1,0){20}}
\put(105,0){\framebox(145,16){
  $\sym{len}^\sharp(L) \to \sym{foldl}^\sharp(\lambda xy.\sym{s}(x),0,L)$}}
\put(250,8){\line(1,0){20}}
\put(270,8){\vector(0,1){22}}

\end{picture}
\end{center}
\caption{The static dependency graph of $R_{\sym{ave}}$}\label{fig:SDG}
\end{figure}
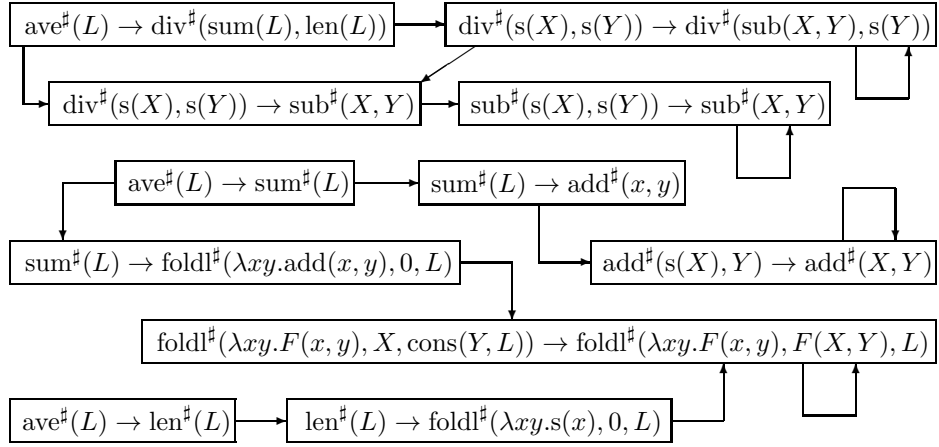
\vspace*{-1pt}
\begin{example} \label{ex:SRC}
For the PFP-HRS $R_{\sym{ave}}$ in Example \ref{ex:SDP},
the static dependency graph of $R_{\sym{ave}}$ is shown in Fig. \ref{fig:SDG}.
Then the set $SRC(R_\sym{ave})$ consists of
the following four static recursion components:
\[\begin{array}{l}
  \left\{\begin{array}{l}
    \sym{foldl}^\sharp(\lambda xy.F(x,y), X, \sym{cons}(Y,L))
      \to \sym{foldl}^\sharp(\lambda xy.F(x,y), F(X,Y), L)
  \end{array}\right\}\\
  \left\{\begin{array}{l}
    \sym{add}^\sharp(\sym{s}(X),Y) \to \sym{add}^\sharp(X,Y)
  \end{array}\right\}\\
  \left\{\begin{array}{l}
    \sym{sub}^\sharp(\sym{s}(X),\sym{s}(Y)) \to \sym{sub}^\sharp(X,Y)
  \end{array}\right\}\\
  \left\{\begin{array}{l}
    \sym{div}^\sharp(\sym{s}(X),\sym{s}(Y)) \to \sym{div}^\sharp(\sym{sub}(X,Y),\sym{s}(Y)
)
  \end{array}\right\}
\end{array}\]
\end{example}


In order to prove the non-loopingness of components,
the notions of subterm criterion and reduction pair have been proposed.
The subterm criterion was introduced on TRSs \cite{HM07},
and then extended to STRSs \cite{KS07} and HRSs \cite{KISB09}.
Reduction pairs \cite{KNT06}
are an abstraction of the notion of weak-reduction order \cite{AG00}.

\begin{definition}[Subterm criterion \cite{KISB09}]
Let $R$ be an HRS and $C \in SRC(R)$.
We say that $C$ satisfies the {\em subterm criterion}
 if there exists a function $\pi$
 from $\D_R^\sharp$ to non-empty sequences of positive integers such that:
\begin{itemize}
\item
  $u|_{\pi(top(u^\sharp))} \rhd_\sub v|_{\pi(top(v^\sharp))}$
  for some $u^\sharp \to v^\sharp \in C$,
\item
  and the following conditions hold for every $u^\sharp \to v^\sharp \in C$:
  \begin{itemize}
  \item
    $u|_{\pi(top(u^\sharp))} \unrhd_\sub v|_{\pi(top(v^\sharp))}$,
  \item
    $\forall p \prec \pi(top(u^\sharp)). top(u|_p) \notin FV(u)$,
  \item
    and $\forall q \prec \pi(top(v^\sharp)). q = \varepsilon \lor top(v|_q) \notin FV(v) \cup \D_R$.
  \end{itemize}
\end{itemize}
\end{definition}


\begin{definition}[Reduction pair, Weak reduction order \cite{AG00,KNT06}]
A pair $(\gsim,>)$ of relations
is a {\em reduction pair}
if $\gsim$ and $>$ satisfy the following properties:
\begin{itemize}
\item
  $>$ is well-founded and closed under substitutions,
\item
  $\gsim$ is closed under contexts and substitutions,
\item
  and ${\gsim \cdot >} \subseteq {>}$ or ${> \cdot \gsim} \subseteq {>}$.
\end{itemize}
In particular,
$\gsim$ is a {\em weak reduction order}
if $(\gsim,\gsim \setminus \lsim)$ is a reduction pair.
\end{definition}

\comment{It is often required that $\gsim$ is a quasi-order
and $>$ is a strict order for a reduction pair $(\gsim,>)$.
Although this definition does not require them,
the pair $(\gsim^*,>^+)$ fulfills the requirements.}

\begin{proposition}\cite{KISB09}\label{pr:SDP2}
Let $R$ be a PFP-HRS
such that there exists no infinite path in the static dependency graph.
Then, $C\in SRC(R)$ is non-looping
if $C$ satisfies one of the following properties:
\begin{itemize}
\item
  $C$ satisfies the subterm criterion.
\item
  There is a reduction pair $(\gsim,>)$ such that
  $R \subseteq {\gsim}$,
  $C \subseteq {{\gsim} \cup {>}}$ and
  ${{C} \cap {>}} \neq \emptyset$.
\end{itemize}
\end{proposition}

\begin{example} \label{ex:subterm}
Let $\pi(\sym{foldl}^\sharp) = 3$
and $\pi(\sym{add}^\sharp) = \pi(\sym{sub}^\sharp) = 1$.
Then, every static recursion component $C$ except the one for $\sym{div}$
(cf. Example \ref{ex:SRC})
satisfies the subterm criterion in the underlined positions below.
Hence, these static recursion components are non-looping.
\[\begin{array}{l}
  \left\{
    \sym{foldl}^\sharp(\lambda xy.F(x,y), X, \ul{\sym{cons}(Y,L)})
      \to \sym{foldl}^\sharp(\lambda xy.F(x,y), F(X,Y), \ul{L})
  \right\}\\
  \left\{
    \sym{add}^\sharp(\ul{\sym{s}(X)},Y) \to \sym{add}^\sharp(\ul{X},Y)
  \right\}
  ~~~
  \left\{
    \sym{sub}^\sharp(\ul{\sym{s}(X)},\sym{s}(Y)) \to \sym{sub}^\sharp(\ul{X},Y)
  \right\}
\end{array}\]
\end{example}

~\section{Argument Filterings}

An argument filtering generates a weak reduction order
from an arbitrary reduction order.
The method was first proposed on TRSs \cite{AG00},
and then extended to STRSs \cite{K01,KS09}.
Since this extension has the problem
that this method may destroy the well-typedness of terms,
Kusakari \REVISED{and Sakai} improved the method
so that the well-typedness is never destroyed \cite{KS09}.
In this section, we expand this technique to HRSs.

\begin{definition}
An {\em argument filtering function} is a function $\pi$
such that, for every $f \in \Sigma$ of type
$\alpha_1 \to \cdots \to \alpha_n \to \beta$ with $\beta \in \B$,
$\pi(f)$ is
either a positive integer $i\le n$ if $\alpha_i=\beta$,
or a list of positive integers $[i_1,\ldots,i_k]$
with $i_1,\ldots,i_k \leq n$.
Then, we extend the function $\pi$ to terms by taking:

\[\pi(\lambda \ol{x_m}.a(\ol{t_n})) \equiv
  \left\{\begin{array}{ll}
    \lambda \ol{x_m}. \pi(t_i)
      & \mbox{ if $a \in \Sigma$ and $\pi(a)=i$} \\
    \lambda \ol{x_m}. a(\pi(t_{i_1}),\ldots,\pi(t_{i_k}))
      & \mbox{ if $a \in \Sigma$ and $\pi(a)=[i_1,\ldots,i_k]$} \\
    \lambda \ol{x_m}. a(\pi(t_1),\ldots,\pi(t_n))
      & \mbox{ if $a \in \V$}
  \end{array}\right.\]

\ \\
Given an argument filtering $\pi$ and a binary relation $>$,
we define $s \gsim_\pi t$ by $\pi(s) > \pi(t)$ or $\pi(s) \equiv \pi(t)$,
and $s >_\pi t$ by $\pi(s) > \pi(t)$.
We also define the substitution $\theta_\pi$
by $\theta_\pi(x) \equiv \pi(\theta(x))$.
Finally, we define the typing function $\type_\pi$ after argument filtering
as $\type_\pi(a) = \alpha_{i_1} \to \cdots \to \alpha_{i_k} \to \beta$
if $a \in \Sigma$, $\pi(a)=[i_1,\ldots,i_k]$,
$\type(a) = \alpha_1 \to \cdots \alpha_n \to \beta$ and $\beta \in \B$;
otherwise $\type_\pi(a) = \type(a)$.
\end{definition}

In the examples, except stated otherwise, $\pi(f)=[1,\ldots,n]$ if
$\type(f) = \alpha_1 \to \cdots \to \alpha_n \to \beta$ and $\beta \in
\B$ (no argument is removed).

For instance,
if $\pi(\sym{sub}) = [1]$
then
\(\pi(\sym{div}^\sharp(\sym{sub}(X,Y),\sym{s}(Y)))
  \equiv \sym{div}^\sharp(\sym{sub}(X),\sym{s}(Y))\).


Note that
our argument filtering method never destroys the well-typedness,
which is easily proved by induction on terms.

\begin{theorem}\label{th:af-welltyped}
For any argument filtering $\pi$ and term $t \in \T$,
$\pi(t)$ is well-typed under the typing function $\type_\pi$
and $\type_\pi(\pi(t))=\type(t)$.
\end{theorem}

In the following,
we prove the soundness of the argument filtering method
as a generating method of weak reduction orders.
To this end, we first prove a lemma required for showing that
$>_\pi$ and $\gsim_\pi$ are closed under substitution.


\begin{lemma} \label{lm:af}
$\pi(t\theta\da) \equiv \pi(t)\theta_\pi\da$.
\end{lemma}

\begin{proof}
\REVISED{%
We proceed by induction on preterm $t\theta$
ordered with ${\red{\beta}{}} \cup {\rhd_\sub}$.
\begin{itemize}
\item
  In case of $t \equiv \lambda x.u$:
  Since $t\theta \rhd_\sub u\theta$,
  we have $\pi(u\theta\da) \equiv \pi(u)\theta_\pi\da$
  from the induction hypothesis.
  Hence we have:
  \(\pi((\lambda x. u)\theta\da)
      \equiv \lambda x. \pi(u\theta\da)
      \equiv \lambda x. \pi(u)\theta_\pi\da
      \equiv \pi(\lambda x.u)\theta_\pi\da \).
\item
  In case of $t \equiv f(\ol{t_n})$, $f \in \Sigma$, and $\pi(f)=i$:
  Since $t\theta \rhd_\sub t_i\theta$,
  we have $\forall i.\,\pi(t_i\theta\da) \equiv \pi(t_i)\theta_\pi\da$
  from the induction hypothesis.
  Hence we have:
  \(\pi(f(\ol{t_n})\theta\da)
      \equiv \pi(f(\ol{t_n\theta\da}))
      \equiv \pi(t_i\theta\da)
      \equiv \pi(t_i)\theta_\pi\da
      \equiv \pi(f(\ol{t_n}))\theta_\pi\da\).
\item
  In case of $t \equiv f(\ol{t_n})$ , $f \in \Sigma$, and $\pi(f)$ is a list:
  Suppose that
  $t'_i \equiv \bot\da$ if $i \notin \pi(f)$;
  otherwise $t'_i \equiv \pi(t_i)$,
  and
  $t''_i \equiv \bot\da$ if $i \notin \pi(f)$;
  otherwise $t''_i \equiv \pi(t_i\theta\da)$.
  For each $i$,
  since $t\theta \rhd_\sub t_i\theta$,
  we have $\pi(t_i\theta\da) \equiv \pi(t_i)\theta_\pi\da$
  from the induction hypothesis.
  Then $t''_i \equiv t'_i\theta_\pi\da$ holds for each $i$.
  Hence we have:
  \(\pi(f(\ol{t_n})\theta\da)
    \equiv \pi(f(\ol{t_n\theta\da}))
    \equiv f(\ol{t''_n})
    \equiv f(\ol{t'_n\theta_\pi\da})
    \equiv f(\ol{t'_n})\theta_\pi\da
    \equiv \pi(f(\ol{t_n}))\theta_\pi\da\).
\item
  In case of $t \equiv X \in \V$:
  Obvious from the definition of $\theta_\pi$.
\item
  In case of $t \equiv X(\ol{t_n})$, $X \in \V$ and $n>0$:
  Since $\type(X) = \type(X\theta)$,
  we have $X\theta \equiv \lambda \ol{y_n}.a(\ol{u_k})$.
  For each $i$,
  since $t\theta \rhd_\sub t_i\theta$,
  we have $\pi(t_i\theta\da) \equiv \pi(t_i)\theta_\pi\da$
  from the induction hypothesis.
  Since
  \(t\theta
      \equiv (\lambda \ol{y_n}.a(\ol{u_k}))(\ol{t_n\theta})
      \red{\beta}{+}
        a(\ol{u_k}) \{\, y_i := t_i\theta\da \mid i \in \ol{n} \,\}\),
  we have
  \(\pi(a(\ol{u_k}) \{ y_i := t_i\theta\da \mid i \in \ol{n} \} \da)
    \equiv
      \pi(a(\ol{u_k})) \{ y_i := \pi(t_i\theta\da) \mid i \in \ol{n} \} \da\)
  from the induction hypothesis.
  Hence we have:
  \(\pi(X(\ol{t_n})\theta\da)
      \equiv
        \pi((\lambda \ol{y_n}.a(\ol{u_k}))(\ol{t_n\theta\da})\da)
      \equiv
        \pi(a(\ol{u_k}) \{\, y_i := t_i\theta\da \mid i \in \ol{n} \,\} \da)
      \equiv
        \pi(a(\ol{u_k})) \{\, y_i := \pi(t_i\theta\da) \mid i \in \ol{n} \,\} \da
      \equiv
        \pi(a(\ol{u_k})) \{\, y_i := \pi(t_i)\theta_\pi\da \mid i \in \ol{n} \,\} \da
      \equiv
        (\lambda \ol{y_n}.\pi(a(\ol{u_k}))) (\ol{\pi(t_n)\theta_\pi\da})\da
      \equiv
        \pi(\lambda \ol{y_n}.a(\ol{u_k}))(\ol{\pi(t_n)\theta_\pi\da})\da\\
      \equiv
        X(\ol{\pi(t_n)})\theta_\pi\da
      \equiv
        \pi(X(\ol{t_n}))\theta_\pi\da\).
\qed
\end{itemize}
}
\end{proof}

\REVISED{
Note that the corresponding lemma in STRSs
is $\pi(t\theta) \geq \pi(t)\theta_\pi$ where $>$ is a given binary relation \cite{KS09}.
This is the technical reason why the argument filtering method on STRSs
can apply to only left-firmness
(left-hand side variables occurs at leaf positions only)
STRSs\cite{K01,KS09}.
This difference originates the fact that
STRSs allow partial application (ex. $\sym{foldl}~F$, $\sym{foldl}~F~X$)
but HRSs does not.
}


\begin{theorem}
For any reduction order $>$ and argument filtering function $\pi$,
$\gsim_\pi$ is a weak reduction order.
\end{theorem}

\begin{proof}
It is easily shown that
$s \gsim_\pi t \imply C[s] \gsim_\pi C[t]$
by induction on $C[\,]$.
From Lemma \ref{lm:af},
we have
\(s \gsim_\pi t
  \imply \pi(s) \geq \pi(t)
  \imply \pi(s)\theta_\pi\da \geq \pi(t)\theta_\pi\da
  \imply \pi(s\theta\da) \geq \pi(t\theta\da)
  \imply s\theta\da \gsim_\pi t\theta\da\),
and
\(s >_\pi t
  \imply \pi(s) > \pi(t)
  \imply \pi(s)\theta_\pi\da > \pi(t)\theta_\pi\da
  \imply \pi(s\theta\da) > \pi(t\theta\da)
  \imply s\theta\da >_\pi t\theta\da\).
Remaining properties are routine.
\qed
\end{proof}


~\begin{example} \label{ex:AF}
Consider the PFP-HRS $R_{\sym{ave}}$ in Example \ref{ex:SDP}.
Every static recursion component except
\(\{\sym{div}^\sharp(\sym{s}(X),\sym{s}(Y)) \to \sym{div}^\sharp(\sym{sub}(X,Y),\sym{s}(Y))\}\)
is non-looping (cf. Example \ref{ex:subterm}).
We can prove its non-loopingness with the argument filtering method,
by taking $\pi(\sym{sub})=\pi(\sym{div^\sharp})=[1]$,
and the normal higher-order reduction ordering $>_{rhorpo}^n$,
written $(>_{rhorpo})_n$ in \Cite{JR06} defined by:
\begin{itemize}
\item a neutralization level ${\cal L}_f^j=0$ for all symbol
  $f\in\Sigma$ and argument position $j$ (in fact, these parameters
  are relevant for functional arguments only),
\item filtering out all arguments (a notion introduced in \Cite{JR06}
  not to be confused with the argument filtering method) by taking
  ${\cal A}_f^j=\emptyset$ for all $f$ and $j$ (again, these
  parameters are relevant for functional arguments only),
\item a precedence $s_{new} >_{\Sigma_{new}} sub_{new}$ (a symbol
  $f_{new}$ with $f\in\Sigma$ is a new symbol introduced by the
  definition of $>_{rhorpo}^n$ in \Cite{JR06}, with the same type as
  $f$ since neutralization levels are null),
\item a multiset (or lexicographic) status for $\sym{div}^\sharp_{new}$,
\item a quasi-ordering on types reduced to the equality (the strict
  part is well-founded since it is empty, and equality preserves
  functional types).
\end{itemize}

Then we have
\(\pi(\sym{div}^\sharp(\sym{s}(X),\sym{s}(Y)))
    \equiv   \sym{div}^\sharp(\sym{s}(X))
    \,>_{rhorpo}^n\, \sym{div}^\sharp(\sym{sub}(X))
    \equiv   \pi(\sym{div}^\sharp(\sym{sub}(X,Y),\sym{s}(Y)))\),
and ${R_\sym{div}} \subseteq {(\ge_{rhorpo}^n)_\pi}$. For instance,
\(\sym{div}^\sharp(\sym{s}(X))
    \,{>_{rhorpo}^n}\\\sym{div}^\sharp(\sym{sub}(X))\)
since \REVISED{$FN$}\((\sym{div}^\sharp(\sym{s}(X)))\da_\beta
    \,>_{rhorpo}\,\) \REVISED{$FN$}\((\sym{div}^\sharp(\sym{sub}(X)))\da_\beta\)
and, because ${\cal L}_f^j=0$ and ${\cal A}_f^j=\emptyset$,
\REVISED{$FN$}$(f t_1\ldots t_n) = f_{new}$ \REVISED{$FN$}$(t_1)\ldots$ \REVISED{$FN$}$(t_n)$.
From Proposition \ref{pr:SDP2},
the static recursion component for $\sym{div}$ is non-looping,
and $R_\sym{div}$ is terminating.
\end{example}

\section{Usable Rules}

In order to reduce the number of constraints
required for showing the non-loopingness of a component,
the notion of usable rules is widely used.
This notion was introduced on TRSs \cite{GTSF06,HM07}
and then extended to STRSs \cite{SKSSN07,KS09}.
In this section, we extend it to HRSs.
\comment{
In this extension,
we improve the definition for ignoring usability of rules,
through higher-order variables in left hand sides of rules,
which is required in STRSs \cite{SKSSN07}.}


To illustrate the interest of this notion, we start with some example.

\begin{example}\label{ex:usable1}
We consider the data type
$\sym{heap} ::= \sym{leaf} \mid \sym{node}(\sym{nat}, \sym{heap}, \sym{heap})$
and the PFP-HRS $R_{\sym{heap}}$ defined by the following rules:

\REVISED{%
\[\left\{\begin{array}{rcl}
  \sym{add}(0,Y)
    &\to& Y\\
  \sym{add}(\sym{s}(X),Y)
    &\to& \sym{s}(\sym{add}(X,Y)) \\
  \sym{map}(\lambda x. F(x), \sym{nil})
    &\to& \sym{nil}\\
  \sym{map}(\lambda x. F(x), \sym{cons}(X, L))
    &\to& \sym{cons}(F(X), \sym{map}(\lambda x. F(x). L)) \\
  \sym{merge}(H, \sym{leaf})
    &\to& H\\
  \sym{merge}(\sym{leaf}, H)
    &\to& H \\
  &&\hspace*{-160pt}
  \sym{merge}(\sym{node}(X_1, H_{11}, H_{12}), \sym{node}(X_2, H_{21}, H_{22}))
    \\&&\hspace*{-50pt}
    \to \sym{node}(X_1, H_{11}, \sym{merge}(H_{12}, \sym{node}(X_2, H_{21}, H_{22}))\\
  &&\hspace*{-160pt}
  \sym{merge}(\sym{node}(X_1, H_{11}, H_{12}), \sym{node}(X_2, H_{21}, H_{22}))
    \\&&\hspace*{-50pt}
    \to \sym{node}(X_2, \sym{merge}(\sym{node}(X_1, H_{11}, H_{12}), H_{21}), H_{22}) \\
  \sym{foldT}(\lambda xyz. F(x,y,z), X, \sym{leaf})
    &\to& X \\
  &&\hspace*{-160pt}
  \sym{foldT}(\lambda xyz. F(x,y,z), X, \sym{node}(Y, H_1, H_2))
    \\&&\hspace*{-130pt}
    \to F(X, \sym{foldT}(\lambda xyz. F(x,y,z), X, H_1),
             \sym{foldT}(\lambda xyz. F(x,y,z), X, H_2)) \\
  \sym{sumT}(H)
    &\to& \sym{foldT}(\lambda xyz. \sym{add}(x, \sym{add}(y, z)), 0, H) \\
  \sym{hd}(\sym{nil})
    &\to& \sym{leaf}\\
  \sym{hd}(\sym{cons}(X, L))
    &\to& X \\
  \sym{l2t}(\sym{nil})
    &\to& \sym{nil}\\
  \sym{l2t}(\sym{cons}(H, \sym{nil}))
    &\to& \sym{cons}(H,\sym{nil}) \\
  \sym{l2t}(\sym{cons}(H_1, \sym{cons}(H_2, L)))
    &\to& \sym{l2t}(\sym{cons}(\sym{merge}(H_1, H_2), \sym{l2t}(L))) \\
  \sym{list2heap}(L)
    &\to& \sym{hd}(\sym{l2t}(\sym{map}(\lambda x. \sym{node}(x, \sym{leaf}, \sym{leaf}), L
)))
\end{array}\right.\]
}
The static recursion components for $\sym{foldT}$ consists of
\[\{\sym{foldT}^\sharp(\lambda xyz. F(x,y,z), X, \sym{node}(Y, H_1, H_2))
    \to \sym{foldT}(\lambda xyz. F(x,y,z), X, H_i)\}\]
for $i=1,2$, and their union.
By taking $\pi(\sym{foldT}) = 3$,
these components satisfy the subterm criterion.
The static recursion components for $\sym{add}$, $\sym{map}$ and $\sym{merge}$
also satisfy the subterm criterion.
Hence it suffices to show that the following three static recursion components
for $\sym{l2t}$ are non-looping:
\[\begin{array}{l}
  \left\{\begin{array}{l}
    \sym{l2t}^\sharp(\sym{cons}(H_1, \sym{cons}(H_2, L)))
      \to \sym{l2t}^\sharp(\sym{cons}(\sym{merge}(H_1, H_2), \sym{l2t}(L)))
      \hspace*{2pt}\cdots (1)
  \end{array}\right\}\\
  \left\{\begin{array}{l}
    \sym{l2t}^\sharp(\sym{cons}(H_1, \sym{cons}(H_2, L)))
      \to \sym{l2t}^\sharp(L)
      \hspace*{2pt}\cdots (2)
  \end{array}\right\}\\
  \left\{ (1), (2) \right\}
\end{array}\]
The component $\{(2)\}$ satisfies the subterm criterion.
By taking $\pi(\sym{cons}) = [2]$ and $\pi(\sym{l2t}) = \pi(\sym{l2t}^\sharp) = 1$,
we can orient the static dependency pairs $(1)$ and $(2)$
by using the normal higher-order recursive path ordering \cite{JR06}:
\[\begin{array}{l}
  \pi(\sym{l2t}^\sharp(\sym{cons}(H_1, \sym{cons}(H_2, L))))
    \\\hspace*{20pt}
    \equiv \sym{cons}(\sym{cons}(L))
    >_{rhorpo}^n \sym{cons}(L)
    \equiv \pi(\sym{l2t}^\sharp(\sym{cons}(\sym{merge}(H_1, H_2), \sym{l2t}(L)))) \\
  \pi(\sym{l2t}^\sharp(\sym{cons}(H_1, \sym{cons}(H_2, L))))
    \equiv \sym{cons}(\sym{cons}(L))
    >_{rhorpo}^n L \equiv \pi(\sym{l2t}^\sharp(L))
\end{array}\]
However, in contrast to Example \ref{ex:AF},
the non-loopingness of $\{(1)\}$ and $\{(1), (2)\}$ cannot be shown
with the previous techniques.
Indeed, we cannot solve the constraint ${R_{\sym{heap}}} \subseteq {\gsim}$.
More precisely,
we cannot orient the rule for $\sym{hd}$,
because
$\pi(\sym{hd}(\sym{cons}(X,L)))\equiv \sym{hd}(\sym{cons}(L))$
does not contain the variable $X$ occurring in the right-hand side.
\end{example}


The notion of usable rule solves this problem,
that is,
it allows us to ignore the rewrite rule for $\sym{hd}$
for showing the non-loopingness of $\sym{l2t}$.

\begin{definition}[Usable rules]
We denote $f >_{\sym{def}} g$
if $g$ is a defined symbol and there is some $l \to r \in R$
such that $top(l) = f$ and $g$ occurs in $r$.

We define the set $\U(t)$ of usable rules of a term $t$ as follows.
If, for every $X(\ol{t_n}) \in Sub(t)$,
$\ol{t_n}$ are distinct bound variables, then
\(\U(t) =
    \{ l \to r \in R \mid
      f >_{\sym{def}}^* top(l)
      \mbox{ for some $f \in \D_R$ occurs in $t$}\}\).
Otherwise, $\U(t)=R$.
The usable rules of a static recursion component $C$
is
\(\U(C) = \bigcup\{ \U(v^\sharp) \mid u^\sharp \to v^\sharp \in C\} \).

For each $\alpha \in \B$,
we associate the new function symbols $\bot_\alpha$ and $\sym{c}_\alpha$
with $\type(\bot_\alpha) = \alpha$
and $\type(\sym{c}_\alpha) = \alpha \to \alpha \to \alpha$.
We define the HRS $C_e$ as 
\(C_e = \{ \sym{c}_\alpha(x_1,x_2) \to x_i \mid \alpha \in \B ,\, i = 1,2 \}\).
\end{definition}

Hereafter we omit the index $\alpha$ whenever no confusion arises.

When we show the non-loopingness of a static recursion component
using a reduction pair,
Proposition \ref{pr:SDP2} requires showing that $R \subseteq {\gsim}$.
The non-loopingness is not guaranteed by simply replacing $R$ with $\U(C)$.
We can supplement the gap with the HRS $C_e$.

\begin{theorem}\label{th:ur}
Let $R$ be a finitely-branching PFP-HRS.
Then $C \in SRC(R)$ is non-looping
if there exists a reduction pair $(\gsim,>)$ such that
  $\U(C) \cup C_e \subseteq {\gsim}$,
  $C \subseteq {\gsim} \cup {>}$, and
  $C \cap {>} \neq \emptyset$.
\end{theorem}

The proof of this theorem will be given at the end of this section.


\begin{example}\label{ex:usable2}
We show the termination of the PFP-HRS $R_{\sym{heap}}$
in Example \ref{ex:usable1}.
We have to show the non-loopingness of the components $\{(1)\}$ and $\{(1), (2)\}$.
To this end, it suffices to show that
the constraint $\U(\{(1), (2)\}) \cup C_e \subseteq {\gsim}$
can be solved
(instead of $R_\sym{heap} \subseteq {\gsim}$).
The usable rules of $\{(1), (2)\}$ are:
\[\left\{\begin{array}{rcl}
  \sym{merge}(H, \sym{leaf})
    &\to& H \\
  \sym{merge}(\sym{leaf}, H)
    &\to& H \\
  &&\hspace*{-133pt}
  \sym{merge}(\sym{node}(X_1, H_{11}, H_{12}), \sym{node}(X_2, H_{21}, H_{22}))
    \\&&\hspace*{-50pt}
    \to \sym{node}(X_1, H_{11}, \sym{merge}(H_{12}, \sym{node}(X_2, H_{21}, H_{22}))\\
  &&\hspace*{-133pt}
  \sym{merge}(\sym{node}(X_1, H_{11}, H_{12}), \sym{node}(X_2, H_{21}, H_{22}))
    \\&&\hspace*{-50pt}
    \to \sym{node}(X_2, \sym{merge}(\sym{node}(X_1, H_{11}, H_{12}), H_{21}), H_{22}) \\
  \sym{l2t}(\sym{nil})
    &\to& \sym{nil} \\
  \sym{l2t}(\sym{cons}(H, \sym{nil}))
    &\to& \sym{cons}(H,\sym{nil}) \\
  \sym{l2t}(\sym{cons}(H_1, \sym{cons}(H_2, L)))
    &\to& \sym{l2t}(\sym{cons}(\sym{merge}(H_1, H_2), \sym{l2t}(L)))
\end{array}\right.\]
The weak reduction order $(>_{rhorpo}^n)_\pi$ orient the rules.
Since $C_e \subseteq (>_{rhorpo}^n)_\pi$,
we conclude that $R_{\sym{heap}}$ is terminating.
\end{example}


In the rest of this section,
we present a proof of Theorem \ref{th:ur}.
We assume that $R$ is a finitely-branching PFP-HRS,
$C$ is a static recursion component of $R$,
and $\Delta = \{ top(l) \mid l \to r \in R \setminus \U(C) \}$.

The key idea of the proof is to use the following interpretation $I$.

Thanks to the Well-ordering theorem,
we assume that every non-empty set of terms $T$ has 
a least element $\least(T)$.

\begin{definition}
For a terminating term $t \in \T_\alpha$, $I(t)$ is defined as follows:
\[I(t) \equiv \left\{\begin{array}{ll}
    \lambda x. I(t')
      & \mbox{if $t \equiv \lambda x. t'$} \\
    a(\ol{I(t_n)})
      & \mbox{if $t \equiv a(\ol{t_n})$ and $a \notin \Delta$} \\
    \sym{c}_\alpha(a(\ol{I(t_n)}), Red_\alpha(\{ I(t') \mid t \red{R\setminus \U(C)}{} t' \}))
      & \mbox{if $t \equiv a(\ol{t_n})$ and $a \in \Delta$}
\end{array}\right.\]
Here, for each $\alpha \in \B$,
$Red_\alpha(T)$ is defined as 
$\bot_\alpha$ if $T = \emptyset$;
otherwise
$\sym{c}_\alpha(u, Red_\alpha(T \setminus \{ u \} ))$
where $u \equiv \least(T)$.
We also define $\theta^I$ by $\theta^I(x) \equiv I(\theta(x))$
for a terminating substitution $\theta$.
\end{definition}

The interpretation $I$ is inductively defined on terminating terms
with respect to ${\rhd_\sub} \cup {\red{R}{}}$,
which is well-founded on terminating terms.
Moreover, the set $\{ I(t') \mid t \red{R}{} t' \}$ is finite because R is finitely branching.
Hence, the above definition of $I$ is well-defined.
As for argument filterings (Theorem \ref{th:af-welltyped}),
this interpretation never destroys well-typedness.

\begin{theorem}
For any terminating $t$,
$I(t)$ is well-typed and $\type(I(t))=\type(t)$.
\end{theorem}

\begin{proof}
It can be easily proved by induction on $t$
ordered by ${\rhd_\sub} \cup {\red{R}{}}$.
\qed
\end{proof}


\begin{lemma}\label{lm:ur1}
Let $t$ be a term and $\theta$ be a substitution such that $t\theta\da$ is terminating.
Then, $I(t\theta\da) \red{C_e}{*} I(t)\theta^I\da \red{C_e}{*} t\theta^I\da$.
\end{lemma}

\begin{proof}
\REVISED{
We prove the claim by induction on
$(\{ \type(x) \mid x \in dom(\theta) \}, t)$
ordered by
the lexicographic combination of
the multiset extension $\rhd_s^{mul}$ of $\rhd_s$,
and ${\rhd_\sub} \cup {\red{R}{}}$.
\begin{itemize}
\item
  In case of $t \equiv \lambda x.t'$:
  Since $t \rhd_\sub t'$,
  we have
  $I(t'\theta\da) \red{C_e}{*} I(t')\theta^I\da \red{C_e}{*} t'\theta^I\da$
  from the induction hypothesis.
  Hence we have:
  \(I((\lambda x.t')\theta\da)
      \equiv I(\lambda x.t'\theta\da)
      \equiv \lambda x.I(t'\theta\da)
      \red{C_e}{*} \lambda x.I(t')\theta^I\da
      \equiv I(\lambda x.t')\theta^I\da\),
  and
  \(I(\lambda x.t')\theta^I\da
      \equiv \lambda x.I(t')\theta^I\da
      \red{C_e}{*} \lambda x.t'\theta^I\da
      \equiv (\lambda x.t')\theta^I\da\).
\item
  In case of $t \equiv a(\ol{t_n})$ and $a \notin \Delta \cup dom(\theta)$:
  For each $i$,
  since $t \rhd_\sub t_i$, 
  we have
  $I(t_i\theta\da) \red{C_e}{*} I(t_i)\theta^I\da \red{C_e}{*} t_i\theta^I\da$
  from the induction hypothesis.
  Hence we have:
  \(I(a(\ol{t_n})\theta\da)
      \equiv I(a(\ol{t_n\theta\da}))
      \equiv a(\ol{I(t_n\theta\da)})
      \red{C_e}{*} a(\ol{I(t_n)\theta^I\da})
      \equiv a(\ol{I(t_n)})\theta^I\da
      \equiv I(a(\ol{t_n}))\theta^I\da\),
  and
  \(I(a(\ol{t_n}))\theta^I\da
      \equiv a(\ol{I(t_n)\theta^I\da})
      \red{C_e}{*} a(\ol{t_n\theta^I\da})
      \equiv a(\ol{t_n})\theta^I\da\).
\item
  In case of $t \equiv X \in dom(\theta)$:
  Obvious from the definition of $\theta^I$.
\item
  In case of $t \equiv X(\ol{t_n})$, $X \in dom(\theta)$ and $n>0$:
  Thanks to the general assumption $\type(X) = \type(X\theta)$,
  we let $X\theta \equiv \lambda \ol{y_n}.a(\ol{u_k})$.
  Since
  $\type(X) = \alpha_1 \to \cdots \to \alpha_n \to \beta \rhd_s \alpha_i = \type(y_i)$
  for each $i$,
  we have
  \(I(a(\ol{u_k})\{ y_i := t_i\theta\da \mid i\in \ol{n} \}\da)
      \red{C_e}{*}
      I(a(\ol{u_k}))\{ y_i := I(t_i\theta\da) \mid i\in \ol{n} \}\da\)
  from the induction hypothesis.
  For each $i$,
  since $t \rhd_\sub t_i$, 
  we have
  $I(t_i\theta\da) \red{C_e}{*} I(t_i)\theta^I\da \red{C_e}{*} t_i\theta^I\da$
  from the induction hypothesis.
  Hence, by Theorem 3.9 in \Cite{MN98} (if $s \red{R}{*} t$ and $\theta
  \red{R}{*} \theta'$ then $s\theta\da \red{R}{*} t\theta'\da$),
  \comment{thanks to Proposition \ref{prop:red-subst},} we have:
  \(I(X(\ol{t_n})\theta\da)
      \equiv
        I((\lambda \ol{y_n}.a(\ol{u_k}))(\ol{t_n\theta\da})\da) 
      \equiv
        I(a(\ol{u_k}) \{\, y_i := t_i\theta\da \mid i\in \ol{n} \,\} \da) 
      \red{C_e}{*}
        I(a(\ol{u_k})) \{\, y_i := I(t_i\theta\da) \mid i\in \ol{n} \,\}\da 
      \red{C_e}{*}
        I(a(\ol{u_k}))\{\, y_i := I(t_i)\theta^I\da \mid i\in \ol{n} \,\}\da 
      \equiv\\ (\lambda \ol{y_n}.I(a(\ol{u_k})))(\ol{I(t_n)\theta^I\da})\da
      \equiv X(\ol{I(t_n)})\theta^I\da 
      \equiv I(X(\ol{t_n}))\theta^I\da\),
  and
  \(I(X(\ol{t_n}))\theta^I\da\\
      \equiv X(\ol{I(t_n)})\theta^I\da 
      \equiv (\lambda \ol{y_n}.I(a(\ol{u_k})))(\ol{I(t_n)\theta^I\da})\da 
      \equiv
        I(a(\ol{u_k})) \{\, y_i := I(t_i)\theta^I\da \mid i\in \ol{n} \,\} \da 
      \red{C_e}{*}
        I(a(\ol{u_k})) \{\, y_i := t_i\theta^I\da \mid i\in \ol{n} \,\} \da 
      \equiv \lambda \ol{y_n}.I(a(\ol{u_k}))(\ol{t_n\theta^I\da})\da 
      \equiv I(\lambda \ol{y_n}.a(\ol{u_k}))(\ol{t_n\theta^I\da})\da 
      \equiv X(\ol{t_n})\theta^I\da\).
\item
  In case of $t \equiv f(\ol{t_n})$ and $f \in \Delta$:
  For each $i$,
  since $t \rhd_\sub t_i$, 
  we have
  $I(t_i\theta\da) \red{C_e}{*} I(t_i)\theta^I\da \red{C_e}{*} t_i\theta^I\da$
  from the induction hypothesis.
  For an arbitrary $t''$ such that $t \red{R\setminus \U(C)}{} t''$,
  we have
  $I(t''\theta\da) \red{C_e}{*} I(t'')\theta^I\da \red{C_e}{*} t''\theta^I\da$
  from the induction hypothesis.
  Hence we have:
  \(I(f(\ol{t_n})\theta\da)
      \equiv
        I(f(\ol{t_n\theta\da})) 
      \equiv
        \sym{c}(f(\ol{I(t_n\theta\da)}),
                Red(\{\, I(t') \mid t\theta\da  \red{}{} t' \,\}))\da 
      \red{C_e}{*} 
        \sym{c}(f(\ol{I(t_n\theta\da)}),
                Red(\{\, I(t''\theta\da) \mid t \red{}{} t'' \,\}))\da 
      \red{C_e}{*}
        \sym{c}(f(\ol{I(t_n)\theta^I\da}),
                Red(\{\, I(t'')\theta^I\da \mid t \red{}{} t'' \,\}))\da 
      \equiv
        \sym{c}(f(\ol{I(t_n)}),\\
                Red(\{\, I(t'') \mid t \red{}{} t'' \,\}))\theta^I\da 
      \equiv I(f(\ol{t_n}))\theta^I\da\),
  and
  \(I(f(\ol{t_n}))\theta^I\da
      \equiv
        \sym{c}(f(\ol{I(t_n)}),\\
                Red(\{\, I(t'') \mid t \red{}{} t'' \,\}))\theta^I\da 
      \equiv
        \sym{c}(f(\ol{I(t_n)\theta^I\da}),
                Red(\{\, I(t'')\theta^I\da \mid t \red{}{} t'' \,\}))\da\\
      \red{C_e}{} 
        f(\ol{I(t_n)\theta^I\da})
      \red{C_e}{*} f(\ol{t_n\theta^I\da})
      \equiv f(\ol{t_n})\theta^I\da\).
\qed
\end{itemize}
}
\end{proof}

For the proof of Theorem \ref{th:ur},
it is enough to show that $I(t\theta\da) \red{C_e}{*} t\theta^I\da$.
In fact,
the corresponding lemma for STRSs was the claim \cite{SKSSN07}.
However,
the proof of the previous lemma required the stronger claim
$I(t\theta\da) \red{C_e}{*} I(t)\theta^I\da \red{C_e}{*} t\theta^I\da$
for applying the induction hypothesis.


\begin{lemma}\label{lm:permutation}
Let $t$ be a term and $\theta$ be a permutation
such that $t\theta\da$ is terminating.
Then, $I(t\theta\da) \equiv I(t)\theta^I\da$.
\end{lemma}

\begin{proof}
\REVISED{
We prove the claim by induction on $t$
ordered by ${\rhd_\sub} \cup {\red{R}{}}$.
\begin{itemize}
\item
  In case of $t \equiv \lambda x.t'$:
  Since $t \rhd_\sub t'$,
  we have $I(t'\theta\da) \equiv I(t')\theta^I\da$
  from the induction hypothesis.
  Hence we have:
  \(I((\lambda x.t')\theta\da)
      \equiv I(\lambda x.t'\theta\da) 
      \equiv \lambda x.I(t'\theta\da) 
      \equiv \lambda x.I(t')\theta^I\da 
      \equiv I(\lambda x.t')\theta^I\da\).
\item
  In case of $t \equiv a(\ol{t_n})$ and $a \notin \Delta \cup dom(\theta)$:
  For each $i$,
  since $t \rhd_\sub t_i$, 
  we have $I(t_i\theta\da) \equiv I(t_i)\theta^I\da$
  from the induction hypothesis.
  Hence we have:
  \(I(a(\ol{t_n})\theta\da)
      \equiv I(a(\ol{t_n\theta\da})) 
      \equiv a(\ol{I(t_n\theta\da)}) 
      \equiv a(\ol{I(t_n)\theta^I\da}) 
      \equiv I(a(\ol{t_n}))\theta^I\da\).
\item
  In case of $t \equiv X(\ol{t_n})$ and $X \in dom(\theta)$:
  Since $\theta$ is a permutation,
  we let $X\theta\da \equiv X'\da$ for a variable $X'$.
  For each $i$,
  since $t \rhd_\sub t_i$, 
  we have $I(t_i\theta\da) \equiv I(t_i)\theta^I\da$
  from the induction hypothesis.
  Hence we have:
  \(I(X(\ol{t_n})\theta\da)
      \equiv I(X'(\ol{t_n\theta\da})) 
      \equiv X'(\ol{I(t_n\theta\da)}) 
      \equiv X'(\ol{I(t_n)\theta^I\da}) 
      \equiv X(\ol{I(t_n)})\theta^I\da 
      \equiv I(X(\ol{t_n}))\theta^I\da\).
\item
  In case of $t \equiv f(\ol{t_n})$ and $f \in \Delta$:
  For each $i$,
  since $t \rhd_\sub t_i$, 
  we have
  $I(t_i\theta\da) \equiv I(t_i)\theta^I\da$
  from the induction hypothesis.
  For an arbitrary $t''$ such that $t \red{R\setminus \U(C)}{} t''$,
  we have
  $I(t''\theta\da) \equiv I(t'')\theta^I\da$
  from the induction hypothesis.
  Since $\theta$ is a permutation,
  we have
  \(\{\, I(t') \mid t\theta\da  \red{}{} t' \,\}
    = \{\, I(t''\theta\da) \mid t \red{}{} t'' \,\}\).
  Hence we have:
  \(I(f(\ol{t_n})\theta\da)
      \equiv
        I(f(\ol{t_n\theta\da})) 
      \equiv
        \sym{c}(f(\ol{I(t_n\theta\da)}),
                Red(\{\, I(t') \mid t\theta\da  \red{}{} t' \,\})) 
      \equiv 
        \sym{c}(f(\ol{I(t_n\theta\da)}),\\
                Red(\{\, I(t''\theta\da) \mid t \red{}{} t'' \,\})) 
      \equiv 
        \sym{c}(f(\ol{I(t_n)\theta^I\da}),
                Red(\{\, I(t'')\theta^I\da \mid t \red{}{} t'' \,\})) 
      \equiv
        \sym{c}(f(\ol{I(t_n)}),
                Red(\{\, I(t'') \mid t \red{}{} t'' \,\}))\theta^I\da
      \equiv I(f(\ol{t_n}))\theta^I\da\).
\qed
\end{itemize}
}
\end{proof}


\begin{lemma}\label{lm:ur2}
Let $l \to r \in C \cup \U(C)$ and $\theta$ be a substitution such that $r\theta\da$ is terminating.
Then, $I(r\theta\da) \equiv r\theta^I\da$.
\end{lemma}

\begin{proof}
\REVISED{
We show the stronger property 
$I(t\theta\da) \equiv t\theta^I\da$
for any $l \to r \in C \cup \U(C)$ and $t \in Sub(r)$.
We prove the claim by induction on $t$.
Note that we have no case that
$t \equiv f(\ol{t_n})$ and $f \in \Delta$.
\begin{itemize}
\item
  In case of $t \equiv \lambda x.t'$:
  Since $t \rhd_\sub t'$,
  we have $I(t'\theta\da) \equiv t'\theta^I\da$
  from the induction hypothesis.
  Hence we have:
  \(I((\lambda x.t')\theta\da)
      \equiv I(\lambda x.t'\theta\da) 
      \equiv \lambda x.I(t'\theta\da) 
      \equiv \lambda x.t'\theta^I\da 
      \equiv (\lambda x.t')\theta^I\da\).
\item
  In case of $t \equiv a(\ol{t_n})$ and $a \notin \Delta \cup dom(\theta)$:
  For each $i$,
  since $t \rhd_\sub t_i$, 
  we have $I(t_i\theta\da) \equiv t_i\theta^I\da$
  from the induction hypothesis.
  Hence we have:
  \(I(a(\ol{t_n})\theta\da)
      \equiv I(a(\ol{t_n\theta\da})) 
      \equiv a(\ol{I(t_n\theta\da)}) 
      \equiv a(\ol{t_n\theta^I\da}) 
      \equiv a(\ol{t_n})\theta^I\da\).
\item
  In case of $t \equiv X(\ol{t_n})$ and $X \in dom(\theta)$:
  Since $\type(X) = \type(X\theta)$,
  we have $X\theta \equiv \lambda \ol{y_n}.a(\ol{u_k})$.
  For each $i$,
  since $t \rhd_\sub t_i$, 
  we have $I(t_i\theta\da) \equiv t_i\theta^I\da$
  from the induction hypothesis.
  If $t_1,\ldots,t_n$ are mutually distinct bound variables,
  then $\{ y_i := t_i\theta\da \mid i\in \ol{n} \}$ is a permutation,
  and hence it follows from Lemma \ref{lm:permutation} that
  \(I(X(\ol{t_n})\theta\da)
      \equiv I((\lambda \ol{y_n}. a(\ol{u_k})) (\ol{t_n\theta})\da) 
      \equiv I(a(\ol{u_k})\{ y_i := t_i\theta\da \mid i\in \ol{n} \}\da) 
      \equiv I(a(\ol{u_k}))\{ y_i := I(t_i\theta\da) \mid i\in \ol{n} \}\da 
      \equiv I(a(\ol{u_k}))\{ y_i := t_i\theta^I\da \mid i\in \ol{n} \}\da 
      \equiv (\lambda \ol{y_n}. I(a(\ol{u_k}))) (\ol{t_n\theta^I}) \da 
      \equiv I(\lambda \ol{y_n}. a(\ol{u_k})) (\ol{t_n\theta^I}) \da
      \equiv X(\ol{t_n})\theta^I\da\).
  Otherwise,
  $I(X(\ol{t_n})\theta\da) \equiv X(\ol{t_n})\theta\da$
  and $\theta = \theta^I$,
  because of $\Delta = \emptyset$.
\qed
\end{itemize}
}
\end{proof}


\begin{lemma}\label{lm:ur3}
If $s \red{R}{} t$ and $s$ is terminating,
then $I(s) \red{\U(C) \cup C_e}{+} I(t)$.
\end{lemma}

\begin{proof}
\REVISED{
From $s \red{R}{} t$,
there exists a rule $l \to r \in R$, a context $E[\,]$,
and a substitution $\theta$
such that $s \equiv E[l\theta\da]$ and $t \equiv E[r\theta\da]$.
We prove the claim by induction on $E[\,]$.
\begin{itemize}
\item
  In case of $E[\,] \equiv \Box$ and $l \to r \in \U(C)$:
  From Lemma \ref{lm:ur1} and \ref{lm:ur2},
  we have:
  \(I(s)
    \equiv I(l\theta\da)
    \red{C_e}{*} l\theta^I\da
    \red{\U(C)}{} r\theta^I\da
    \equiv I(r\theta\da)
    \equiv I(t) \).
\item
  In case of $E[\,] \equiv \lambda x.E'[\,]$:
  \(I(\lambda x. E'[l\theta\da])
      \equiv \lambda x. I(E'[l\theta\da]) 
      \red{\U(C) \cup C_e}{+}\\ \lambda x.I(E'[r\theta\da]) 
      \equiv I(\lambda x.E'[r\theta\da])\).
\item
  In case of $E[\,] \equiv a(\ldots,E'[\,],\ldots)$ and $a \notin \Delta$:
  \(I(E[l\theta\da])
      \equiv f(..,I(E'[l\theta\da]),..)\\
      \red{\U(C) \cup C_e}{+} f(\ldots,I(E'[r\theta\da]),\ldots) 
      \equiv I(E[r\theta\da])\).
\item
  In case of $s \equiv f(\ol{s_n})$ and $f \in \Delta$:
  \(I(s)
      \equiv
        I(f(\ol{s_n}))
      \equiv
        \sym{c}(f(\ol{I(s_n)}), Red(\{\, I(v) \mid s \red{}{} v \,\})) 
      \red{C_e}{}
        Red(\{\, I(v) \mid s \red{}{} v \,\}) 
      \red{C_e}{+}
        I(t)\).
\qed
\end{itemize}
}
\end{proof}


Finally, we give the proof of the main theorem for usable rules:
\smallskip

\noindent{\bf Proof of Theorem \ref{th:ur}}.
Assume that static dependency pairs in $C$ generate an infinite chain
\(u_0^\sharp \to v_0^\sharp, u_1^\sharp \to v_1^\sharp, \ldots\),
in which every $u^\sharp \to v^\sharp \in C$ occurs infinitely many times.
Then there exist $\theta_0, \theta_1, \theta_2, \ldots$
such that
for each $i$,
$v_i^\sharp\theta_i\da \red{R}{*} u_{i+1}^\sharp\theta_{i+1}\da$.
Let $i$ be an arbitrary number.
From Lemma \ref{lm:ur1}, \ref{lm:ur2} and \ref{lm:ur3},
we have:
\(v_i^\sharp\theta_i^I\da
    \equiv I(v_i^\sharp\theta_i\da)
    \red{\U(C) \cup C_e}{*} I(u_{i+1}^\sharp\theta_{i+1}\da)
    \red{C_e}{*} u_{i+1}^\sharp\theta_{i+1}^I\da\).
Hence we have
\(v_i^\sharp\theta_i^I\da
    \gsim u_{i+1}^\sharp\theta_{i+1}^I\da
    \gsim v_{i+1}^\sharp\theta_{i+1}^I\da \)
from $\U(C) \cup C_e \subseteq {\gsim}$.
Moreover,
from $C \subseteq {\gsim} \cup {>}$ and $C \cap {>} \neq \emptyset$,
we have
$u_j^\sharp\theta_j^I\da > v_j^\sharp\theta_j^I\da$
for infinitely many $j$.
This contradicts the well-foundedness of $>$.
\qed

\section{Conclusion}

By using the notion of accessibility \cite{blanqui02tcs,B00}, we
extended in an important way the class of systems to which the static
dependency pair method \cite{KISB09} can be applied. We then extended
to HRSs some methods initially developed for TRSs: arguments
filterings \cite{AG00} and usable rules \cite{GTSF06,HM07}. So,
together with the subterm criterion for HRSs \cite{KISB09} and the
normal higher-order recursive path ordering \cite{JR06}, this paper
provides a strong theoretical basis for the development of an
efficient automated termination provers for HRSs, since all these
methods have been shown quite successful in the termination
competition on TRSs \cite{tc} and are indeed the basis of current
state-of-the-art termination provers for TRSs
\cite{giesl06ijcar,HM07}. We now plan to implement all these
techniques, all the more so since some competition on the termination
of higher-order rewrite systems is under consideration
\cite{tc-hrs}. Currently, HORPO is the only technique for higher-order
rewrite systems that has been implemented \cite{horpo-prolog}. One
could also build over \Cite{K06,sternagel09tphol,contejean10pepm} to
provide certificates for these techniques in the case of HRSs.

However, there are still some theoretical problems. Currently, the
static dependency pair method does not handle function definitions
involving data type constructors with functional arguments in a
satisfactory way like, for instance, the rule $Sum5$ of Van de Pol's
formulation of $\mu$CRL \cite{vandepol96phd}:
\[\Sigma(\lambda d.Pd)\circ X\to \Sigma(\REVISED{(}\lambda d.Pd\REVISED{)}\circ X)\]
The first reason is that these arguments are not safe (Definition
\ref{def-safe}). This can be fixed by considering a more complex
interpretations for base types \cite{B00}. The second reason is that
it gives rise to the static dependency pair \(\Sigma(\lambda
d.Pd)\circ X\to Pd\circ X\) the right-hand side of which contains a
variable $d$ not occurring in the left-hand side. And, currently, no
technique can prove the non-loopingness of this static recursion
component, a problem occurring also in \Cite{blanqui06wst-hodp}.

\begin{acknowledgment}
\REVISED{We would like to thank the anonymous referees for their helpful comments.}

This research was partially supported by MEXT KAKENHI \#20500008.
\end{acknowledgment}


\end{document}